\documentclass[10pt,reqno]{amsart}
\usepackage[margin=1.5in]{geometry}
\usepackage{hyperref}
\usepackage{float}
\usepackage{mathtools}
\usepackage[all]{xy}

\usepackage{braket}
\usepackage{tikz-cd}
\usepackage{quiver}
\usepackage{citeref}
\usepackage{amscd,amssymb,amsmath,latexsym,bm}
\usepackage[mathcal,mathscr]{euscript}
\usepackage{lipsum}
\usepackage{amsfonts}
\usepackage{graphicx}
\usepackage{epstopdf}
\usepackage{algorithmic}
\usepackage{hyperref}
\usepackage{xcolor}
\usepackage{upgreek}
\usepackage{enumerate}
\usepackage{ulem}
\usepackage{stmaryrd}
\usepackage{gensymb}			
\usepackage[utf8]{inputenc}
 
\newtheorem{theorem}{Theorem}[section]

\newtheorem{proposition}[theorem]{Proposition}
\newtheorem{remark}[theorem]{Remark}
\newtheorem{definition}[theorem]{Definition}
\newtheorem{example}[theorem]{Example}

\numberwithin{equation}{section}
\numberwithin{figure}{section}

\newcommand{\BM}{{\mathbb B}}
\newcommand{\CM}{{\mathbb C}}
\newcommand{\NM}{{\mathbb N}}
\newcommand{\QM}{{\mathbb Q}}
\newcommand{\RM}{{\mathbb R}}

\newcommand{\ZM}{{\mathbb Z}}

\newcommand{\KM}{{\mathbb K}}

\newcommand{\UM}{{\mathbb U}}

\newcommand{\FM}{{\mathbb F}}

\newcommand{\Dd}{{\mathcal D}}

\newcommand{\Tt}{{\mathcal T}}
\newcommand{\Rr}{{\mathcal R}}

\newcommand{\Cc}{{\mathcal C}}

\newcommand{\Ll}{{\mathcal L}}

\newcommand{\Hh}{{\mathcal H}}

\begin{document}
\title[]{Topological Dynamics of Synthetic Molecules}

\author{Yuming Zhu}

\address{Bronx High School of Science, \\
Bronx, NY 10016, USA \\
\href{mailto:zhuy3@bxscience.edu}{zhuy3@bxscience.edu}}

\author{Emil Prodan}

\address{Department of Physics and
\\ Department of Mathematical Sciences 
\\Yeshiva University 
\\New York, NY 10016, USA \\
\href{mailto:prodan@yu.edu}{prodan@yu.edu}}

\date{\today}

\begin{abstract} 
We study the dynamics of synthetic molecules whose architectures are generated by space transformations from a point group acting on seed resonators. We show that the dynamical matrix of any such molecule can be reproduced as the left regular representation of a self-adjoint element from the stabilized group's algebra. Furthermore, we use elements of representation theory and K-theory to rationalize the dynamical features supported by such synthetic molecules up to topological equivalences. These tools enable us to identify a set of fundamental models which generate by superposition all possible dynamical matrices up to homotopy equivalences. Interpolations between these fundamental models give rise to topological spectral flows.
\end{abstract}

\thanks{This work was supported by the U.S. National Science Foundation through the grant CMMI-2131760 and by U.S. Army Research Office through contract W911NF-23-1-0127.}

\maketitle


\setcounter{tocdepth}{1}

\section{Introduction}
\label{Sec:Introduction}

The discovery of topological insulators \cite{Hal1988,KaneMele2005I,KaneMele2005II,BHZ2006,KWB2007,
MooreBalents2007,FuKane2007,HQW2008} and of topological photonic and mechanical systems \cite{HaldaneRaghu2008,PP2009,WangNature2009,KaneNatPhys2013,
HafeziNatPhot2013,NashPNAS2015,WuPRL2015,SusstrunkScience2015}, together with their complete classification \cite{SRFL2008,QiPRB2008,Kit2009,RSFL2010}, have permanently changed how research and discovery are conducted in materials science. For example, the focus has shifted from enhancing the properties of a material to making it different. Furthermore, many applications rely now on robust effects that do not require fine tuning, such as generating wave channels at the interface between two topologically distinct materials, edge-to-edge topological pumping, or closing and opening of the resonant gaps by interpolating between two distinct material configurations. The latter application is called topological spectral engineering and we will see it at work in the present study. Besides supplying the means for important applications, observing a topological spectral flow is the simplest way to assess if two metamaterials are topologically distinct (see {\it e.g} \cite{LuxJMP2024}). These aforementioned robust effects are applicable to information and sensing technologies, where the main function of a device is to switch between on and off states as robustly as possible. 

Physical models in condensed matter physics are developed on periodic lattices $\Ll$ and most of the time they take the form of tight-binding Hamiltonians 
\begin{equation}\label{Eq:H}
H = \sum_{x,y\in \Ll} t_{xy} \otimes |x\rangle \langle y| + t_{xy}^\ast \otimes |y\rangle \langle x|,
\end{equation}
where the intuition is that an electron hops between the orbitals supported by the sites of the lattice, with probabilities encoded in the matrices $t_{xy}$. In particular, the Hamiltonians of the established topological insulators are of this form (see \cite{RSFL2010} for complete list). If one moves away from this context, unless specialized mathematical tools are utilized, it is not clear what a topological classification might be, how to write {\it all} representative Hamiltonians, and how to compute topological invariants. For example, statements like ``two spectrally gappped Hamiltonians are topologically equivalent if we can deform one into the other without closing the gap" are without value if the space in which the deformations occur is not specified. Imposing symmetry constraints on the Hamiltonians is still not enough because locality conditions on the hopping amplitudes are also needed. In a remarkable paper \cite{Bellissard1986}, Jean Bellissard realized that the full set of constraints can be communicated by simply stating that the Hamiltonians belong to a particular $C^\ast$-algebra. The isomorphism class of this $C^\ast$-algebra identifies a specific class of materials and a concrete realization of it supplies the topological space for the deformations. These statements of Bellissard are not purely formal because, in the same paper, he showed how to construct such $C^\ast$-algebras (see also \cite{MeslandJGP2024}). Once the deformation space is identified, then various flavors of stable homotopy theories, aka K-theories, supply the devices that answer the classification problems.

To the experts, the above statements and their implications are very clear, but this is not at all the case for a broader community of materials scientists, despite of many worked examples available in the literature. The authors believe that this current state of affairs is due to the fact that since these mentioned applications targeted challenging situations, there was no room for pedagogical expositions. As a result, the connections between the K-theoretic methods and other more familiar methods have not been revealed explicitly enough. The purpose of this expository paper is to fill in some of the mentioned shortcomings. The plan is to engage simple synthetic systems built from a finite number of self-coupling resonators and whose architecture is generated using the action of a finite subgroup of the full Euclidean group, hence, a point group (see section~\ref{Sec:SyntheticMolecules}). The dynamical matrices determining the dynamical linear regimes of such systems all fall into the group algebra of the point group. As we shall see, we are in a situation where the representation theory of this group and the K-theory of its algebra overlap quite strongly (see section~\ref{Sec:AlgTopDyn}). Thus, we are in a special situation where the K-theoretic tools can be seen at work through the prism of representation theory. For a community familiar with the latter, this class of physical systems and their analysis can be used as a door into the world of the K-theoretic ideas, which divert from representation theory as soon as the number of degrees of freedom become infinite. 

The exposition is organized as follows. Section~\ref{Sec:SyntheticMolecules} describes quantum and classical materials generated from actions of the Euclidean group on seeding synthetic atoms (such as quantum dots) or self-coupling classical resonators (such as solid shapes fitted with magnets), respectively. We use the terms synthetic molecules to refer to both these systems. Using simple physical reasoning, we calculate the algebra of the Hamiltonians. In the case where all the actions come from a subgroup of the Euclidean group, we find that this algebra drops to the group algebra of the space transformations. This, together with a natural topology on the algebra, completes the task of finding the deformation space for this particular class of synthetic molecules. Section~\ref{Sec:AlgTopDyn} explains the interplay between representation theory and K-theory and demonstrates how numerical invariants are computed in general and in particular for setting of synthetic molecules. Section~\ref{Sec:FundModels} builds a set of Hamiltonians derived from K-theoretic data. Any other Hamiltonian can be deformed to a direct sum of Hamiltonians from this list. Furthermore, we demonstrate that interpolations between pairs of Hamiltonians from the list result in topological spectral flows, hence confirming that they are topologically distinct.

\section{Synthetic molecules}
\label{Sec:SyntheticMolecules}

\subsection{Building architectures with space transformations}

\begin{figure}[t]
\begin{center}
\hspace{10cm}
\includegraphics[width=0.8\textwidth]{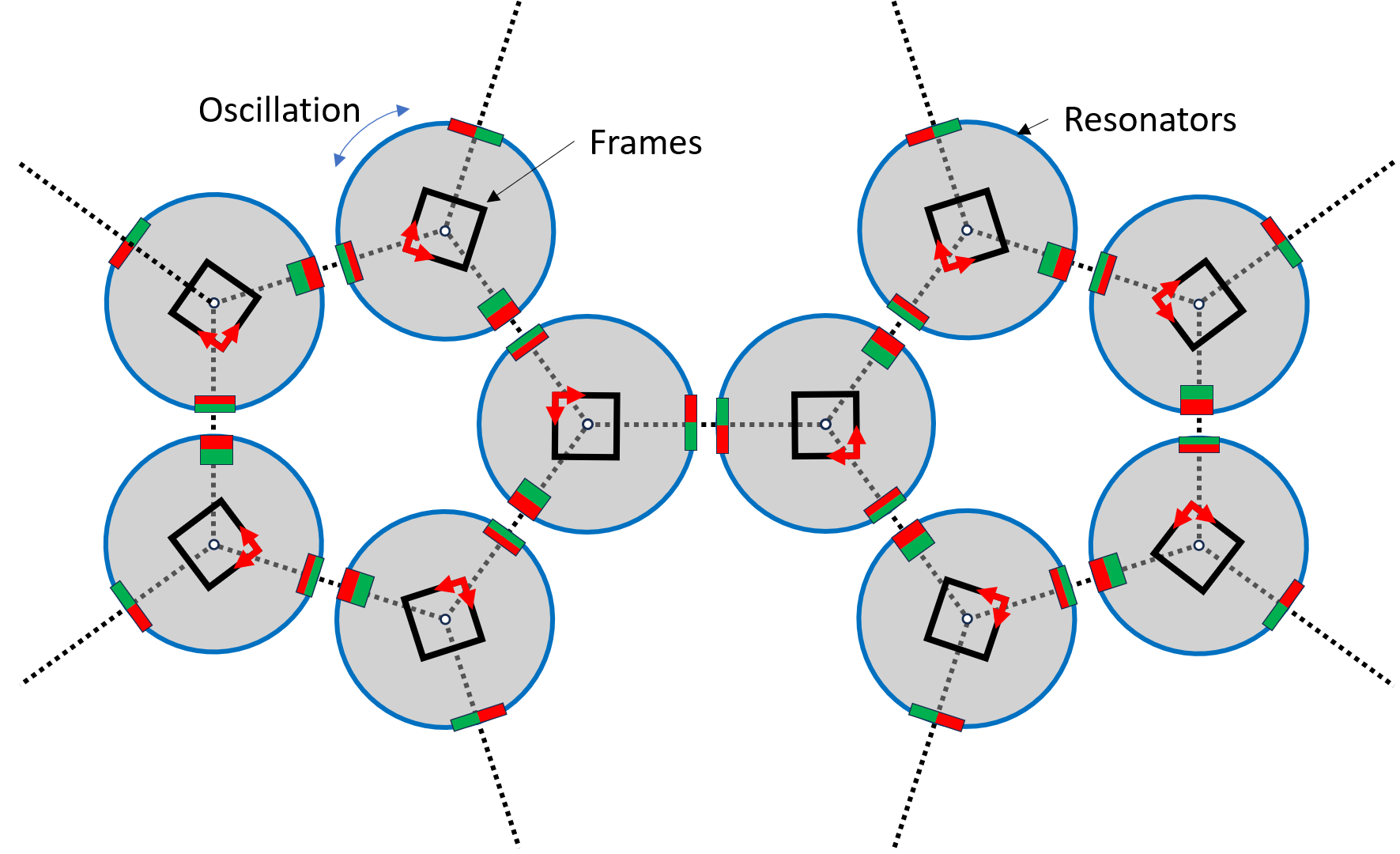}
  \caption{\small A resonator design that leads to a stable equilibrium configuration symmetric to the icosahedron group. 
}
 \label{Fig:SymCp}
\end{center}
\end{figure}

For simplicity, we describe these synthetic materials in a classical mechanical setting, but all our predictions apply in the acoustic, photonic and quantum settings as well. The synthetic mechanical materials that we have in mind are built from many copies of a single seed resonator. The seed resonator consists of moving parts, a physical frame and sources of force fields, such that: 

\begin{itemize}
     \item The frame of the seed resonator can be decorated with sources of force fields, {\it e.g.} magnets, in order to adjust its internal resonant modes and frequencies;
    \item The frame of the seed resonator can also be fitted with sensors in order to quantify the motion of its moving parts. 
    \item The frames of the copies\footnote{A copy of the seed resonator will include the frame, the moving parts and all the other fittings.} can be rigidly anchored once their positions and orientations are decided;
    \item The moving parts of the identical copies of the seed resonators self-couple via the force fields generated by the attached sources of force fields.
\end{itemize}

A schematic example of such synthetic material is shown in Fig.~\ref{Fig:SymCp}, and an actual laboratory model is shown and described in Fig.~\ref{Fig:RealSynthMolecule}. We carefully listed our assumptions because they have several important consequences:
\begin{itemize}
    \item The physical frames and their fitted sensors supply local reference frames, marked by the red arrows in Fig.~\ref{Fig:SymCp}, from where the motion is quantified. Thus, if $\{q_\alpha\}$ are the generalized degrees of freedom of the seed resonator and $q_i(t)$ are their recorded values as function of time in an experiment, then same recordings will be reported if same experiment is carried out with the seed resonator placed at a different location and with different orientation.
    \item The collective motion of the self-coupled moving parts is fully determined by the positions and orientations of the frames.
    \item The motion, as quantified, or better said reported from the local frames, is unaffected by rigid translations and rotations of the whole structure, which is a consequence of Galilean invariance of the physical laws involved in the couplings.
\end{itemize}

We will be dealing only with linear dynamical regimes, which requires the existence of stable equilibrium configurations. Thus, we assume:
\begin{itemize}
    \item For each spatial configuration of the frames of the resonators, there is a unique lowest energy configuration.
\end{itemize}
This assumption needs to be true only for the set of configurations sampled during an experiment or for the configurations engaged during an application. 

\begin{figure}[t]
    \centering
    \includegraphics[width=0.4\linewidth]{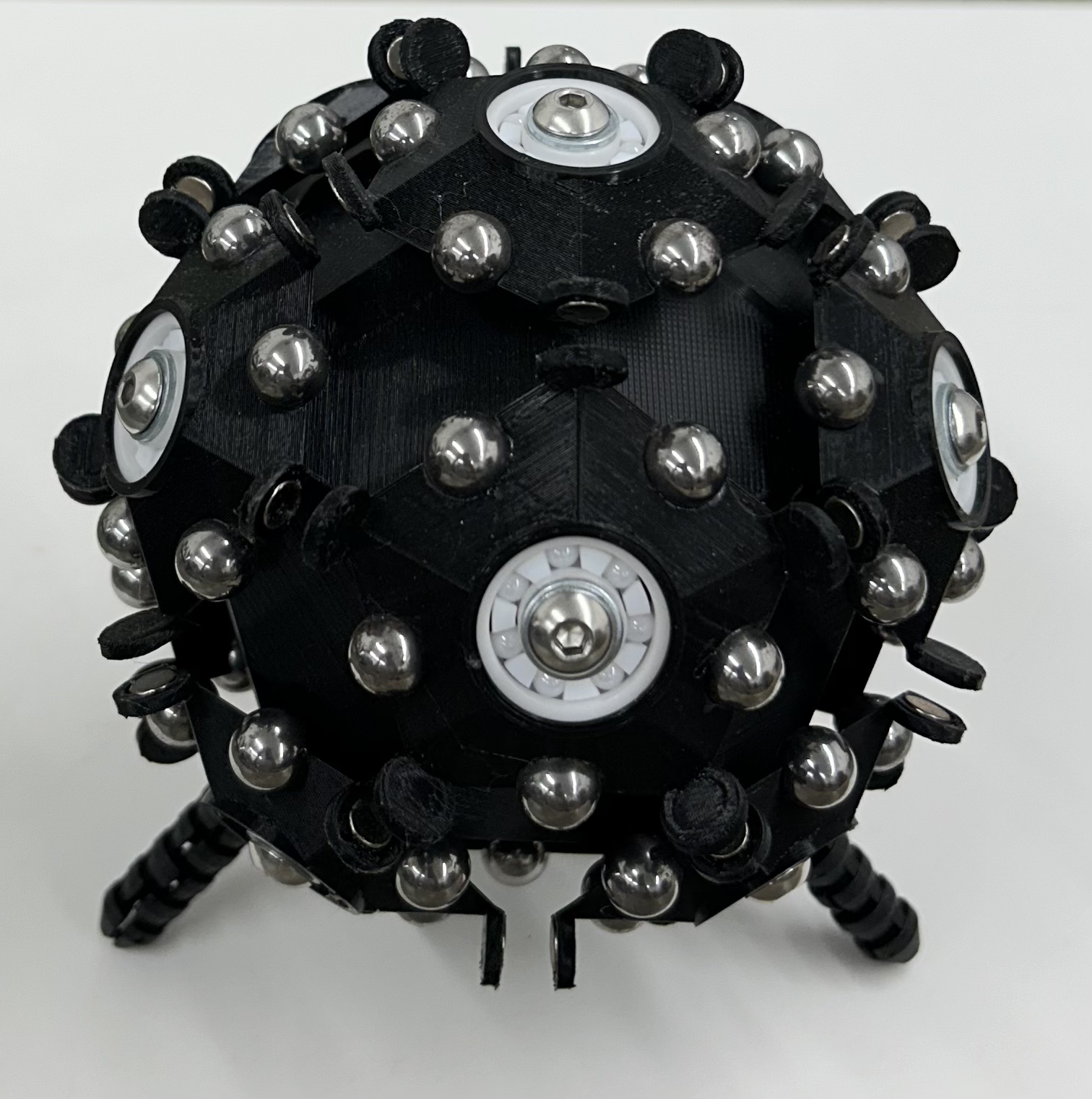}
    \caption{Example of a fully assembled and functional synthetic molecule.}
    \label{Fig:RealSynthMolecule}
\end{figure}

A somewhat trivial observation supplies the key to the operator algebraic approach we want to describe. Specifically, the position and orientation of the frame of a resonator can be set by a rotation followed by a translation of the seed resonator. These two space transformations define an element of the Euclidean group, which transforms the local reference frame of the seed resonator\footnote{All identical copies are generated from a seed resonator that is fixed once and for all at a specific location and orientation in space.} into the local reference frame of an actual resonator. As such, there is a one-to-one correspondence between a resonator and an element of the Euclidean group. We reached our first important conclusion:

\begin{proposition}[\cite{MeslandJGP2024}]\label{Prop:Architecture1}
    Let us use the terms ``synthetic molecules" for clusters of self-coupled resonators and let us define the architecture of the synthetic molecule to be the information contained in the positions and orientations of the frames of the resonators. Then the architecture can be conceptualized as a discrete subset or better said a lattice $\Ll$ of the Euclidean group of space transformations.
\end{proposition}

\begin{remark}{\rm
    The Euclidean group and its subgroup of rotations will play active and passive roles in our discussion. For example, these groups will supply space transformations on $\Ll$, but will also supply the labels (or the coordinates) for the points in $\Ll$. When appearing with an active role, the elements of the groups are indicated by $g$, and otherwise by $x$.
    }$\Diamond$
\end{remark}

We recall that the Euclidean group $\bm E$ is a topological group and that, for the Euclidean space with 3-dimensions, the underlying topological space is
\begin{equation}\label{Eq:ETopo}
    \bm E \simeq \RM^3 \times \RM \mathbb P^3 \times \{-1,1\}.
\end{equation}
Note that $\RM \mathbb P^3 \times \{-1,1\}$ is the topological space underlying the subgroup $O(3)$ of rotations, with $\pm 1$ indicating if a rotation is proper or improper, respectively. According to our statement from above, the architecture of a synthetic molecule containing $N$ resonators can be quantified as a subset of $N$ distinct points on the space~\eqref{Eq:ETopo}. 

\begin{remark}{\rm
    Since we will put strong emphasis on synthetic molecules produced by a finite subgroup of $O(3)$, we supply further details about this group. If a point of the 3-dimensional space is specified by a column vector with three entries encoding its coordinates, then the rotation transformations take the form of a matrix. Any proper rotation takes the form of a $3 \times 3$ matrix of the form
\begin{equation}\label{Eq:Transf}
R(\hat n, \theta) = \exp( \theta \, \hat n \cdot \vec L)
\end{equation}
where $\theta$ is an angle, $\hat n =(n_x,n_y,n_z)$ is 3-component vector of norm one and $\vec L$ is the 3-component vector with entries 
\begin{equation}
L_x = {\small \begin{pmatrix}
0 & 0 & 0 \\ 0 & 0 & -1 \\ 0 & 1 & 0
\end{pmatrix}}, \ 
L_y = {\small \begin{pmatrix}
0 & 0 & 1 \\ 0 & 0 & 0 \\ -1 & 0 & 0
\end{pmatrix}}, \ 
L_z = {\small \begin{pmatrix}
0 & -1 & 0 \\ 1 & 0 & 0 \\ 0 & 0 & 0
\end{pmatrix}}.
\end{equation}
Up to a factor, $\vec L$ coincides with the angular momentum in quantum mechanics for a particle with spin 1. The somewhat complicated space $\RM \mathbb P^3$ is a result of the allowed ranges for the pair $(\theta,\hat n)$ or, equivalently, for the vector $\vec n_\theta = \theta \hat n$. For $\|\vec n_\theta\| < \pi$, \eqref{Eq:Transf} supplies distinct rotations, but for $\|\vec n_\theta\| = \pi$,  \eqref{Eq:Transf} supplies the same rotation if plug in $\vec n_\theta$ and $-\vec n_\theta$. Thus, rotations are fully parameterized by the points of a 3-dimensional ball of radius $\pi$ if the opposite points at its surface are identified. The resulting space is diffeomorphic to $\RM \mathbb P^3$, the space of lines in $\RM^4$ passing through the origin.
}$\Diamond$
\end{remark}

\begin{figure}[t]
    \centering
    \includegraphics[width=0.8\linewidth]{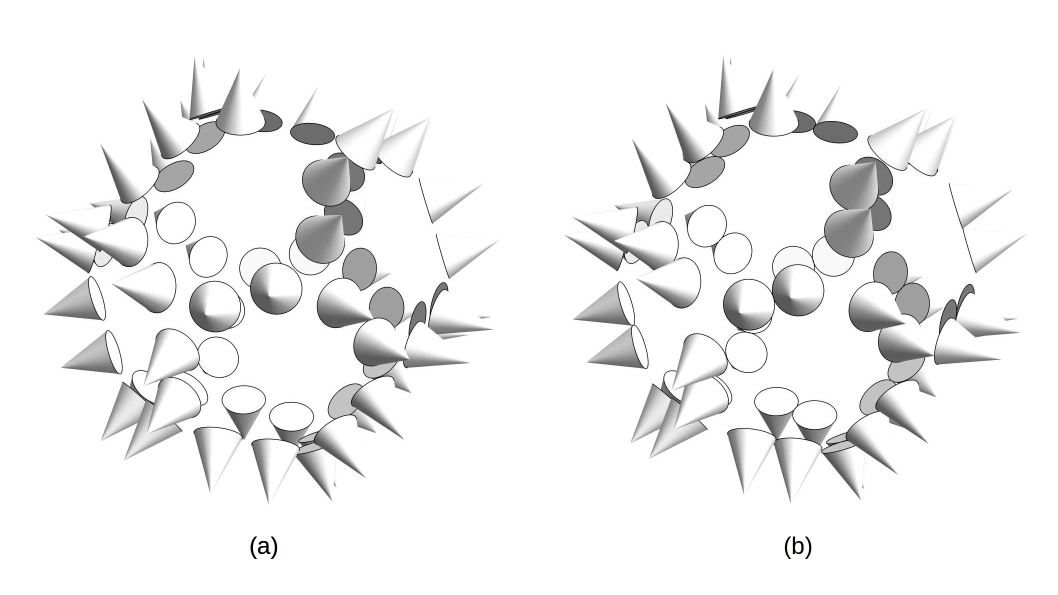}
    \caption{Architectures generated by acting with the proper icosahedral group on a seeding shape. The difference between panels (a) and (b) is the location of the seeding shape.}
    \label{Fig:IHMolecule}
\end{figure}

So far, the discussion was about formalizing the architecture of given material. However, we can reverse the arrows and use Proposition~\ref{Prop:Architecture1} to generate interesting architectures. In Fig.~\ref{Fig:IHMolecule}, we show two architectures generated by acting with the space transformations contained in the proper icosahedral point group on a seed resonator placed at two different locations in space. These and other similar examples will be discussed in detail in subsection~\ref{Sec:PointGroup}.

\subsection{Dynamics} We denote by $G$ the underlying group which labels the resonators of a synthetic molecule. Thus, $G$ can be the full Euclidean group or just its $O(3)$ subgroup. The seed resonator is labeled by the identity $e\in G$ and its degrees of freedom are $q(e,\alpha)$, $\alpha =1,\ldots,D$, with $D$ specifying the number of internal degrees of freedom. Then the degrees of freedom of a resonator labeled by $x \in \Ll \subset G$ in the synthetic molecule are $q(x,\alpha)$, $\alpha=1,\ldots,D$. When the system is driven harmonically with pulsation $\omega$, under the assumption of very small but finite dissipation, the degrees of freedom display an oscillatory behavior as functions of time $t$, 
\begin{equation}
    q(x,\alpha;t) = Q_R(x,\alpha) \cos(\omega t + \phi_{x,\alpha}) = {\rm Re}[Q(x,\alpha) e^{\imath \omega t}],
\end{equation}
where $Q_R \in \RM$ are real-valued amplitudes, and $Q \in \CM$ are complex-valued amplitudes which efficiently incorporate the phases $\phi$. We can place the complex amplitudes in one single vector 
\begin{equation}\label{Eq:QVec}
|Q\rangle = \sum_{x \in \Ll} \sum_{\alpha=1}^D Q(x,\alpha)|x,\alpha\rangle,
\end{equation}
where $|x,\alpha\rangle$ can be thought of as an abstract basis for the state space, or as a concrete column vector with entry 1 at one place and 0 entry at all other places.\footnote{The latter requires an (un-natural) ordering of $(x,\alpha)$'s and this is why the former viewpoint is preferred.} The vectors~\eqref{Eq:QVec} span the Hilbert space $\ell^2(\Ll,\CM^d) \simeq \CM^D \otimes \ell^2(\Ll)$. As a vector in $\CM^D \otimes \ell^2(\Ll)$, the dynamical state of the system takes the form $|Q\rangle = \sum_{x \in \Ll} \bm Q_x \otimes |x \rangle$, with $\bm Q_x \in \CM^d$.

External generalized forces $\{f_{g,\alpha}e^{\imath \omega t}\}$ driving the synthetic molecule can be also encoded as a vector from this Hilbert space, 
\begin{equation}
    |F\rangle = \sum_{x \in \Ll} f_{x,\alpha} |x,\alpha\rangle = \sum_{x\in \Ll} \bm f_x \otimes |x\rangle.
\end{equation} If the quadratic form of the Lagrange function is
\begin{equation}
L(Q_R,\dot Q_R) = \tfrac{1}{2} \dot Q_R^T \bm T \dot Q_R + \tfrac{1}{2} Q_R^T \bm V Q_R,
\end{equation}
where $\bm T$ and $\bm V$ are positive operators on $\CM^D \otimes \ell^2(\Ll)$, then the response of the synthetic molecule to the driving forces is given by the solution to the equation
\begin{equation}
-\omega^2 \bm T |Q\rangle +\bm V |Q\rangle =|F\rangle \ \ \Leftrightarrow \ \ \bm T^{\frac{1}{2}}(\bm D - \omega^2)\bm T^{\frac{1}{2}} |Q \rangle = |F \rangle,
\end{equation}
where 
\begin{equation}
    \bm D : = \bm T^{-\frac{1}{2}} \bm V \bm T^{-\frac{1}{2}}
    \end{equation}
is call the dynamical matrix. The resonant frequencies and the corresponding self-oscillating modes can be obtain from the spectral properties of $\bm D$, which, as any linear operator on the Hilbert space $\CM^D \otimes \ell^2(\Ll)$, it takes the form
\begin{equation}\label{Eq:DMatrix}
 \bm D = \sum_{x,x'\in \Ll} w_{x',x}(\Ll) \otimes |x' \rangle \langle x|,
 \end{equation}
where $w_{x',x}$ are $D \times D$ matrices with complex entries. Often, these are referred to as coupling matrices. Let us point out the similarity between the expressions \eqref{Eq:DMatrix} and \eqref{Eq:H}, where the only difference is that the lattices live in different topological groups, Euclidean group in the former and $\RM^3$ in the latter.

The group $G$ can act on itself from the left or from the right and we will be interested in both. The left action is supplied by multiplication to the left, while the right action goes as $g\cdot x = x g^{-1}$, for all $g,x \in G$. These actions are also well defined on subsets of $G$, in particular, on $\Ll$. Now, in Eq.~\eqref{Eq:DMatrix}, we were careful to specify that the coupling matrices of the synthetic molecule are entirely determined by the architecture of the molecule, hence by the lattice $\Ll$, as it was already explained in the previous subsection. From the principles stated there, we can also infer that the coupling matrices satisfy the following covariant relation,
 \begin{equation}\label{Eq:CovRel}
 w_{g \cdot x',g \cdot x}(g \cdot \Ll) = w_{x',x}(\Ll),
 \end{equation}
 which is a direct consequence of the Galilean invariance of the physical laws underlining the coupling of the resonators. Taking $g = x$, one finds
 \begin{equation}
     w_{x',x}(\Ll) = w_{x \cdot x',e}(x \cdot \Ll),
 \end{equation}
or equivalently,
\begin{equation}\label{Eq:RedDMat}
    \bm D = \sum_{x,x'\in \Ll} w_{x \cdot x'}(x \cdot \Ll) \otimes |x'\rangle \langle x|,
\end{equation}
where we removed the redundant lower index. As shown in \cite{MeslandJGP2024}, this particular form of the dynamical matrices, together with the assumption that the coupling matrices depend continuously on the architecture, enable one to construct the $C^\ast$-algebra where the models live. More precisely, if one keeps the architecture fixed but changes the internal structure of the seed resonator, the dynamical matrices densely sample a groupoid $C^\ast$-algebra that is entirely determined by the lattice $\Ll$. We discuss next how this $C^\ast$-algebra emerges in the particular setting.

\subsection{Synthetic molecules generated with point groups}
\label{Sec:PointGroup}

The finite subgroups of $SO(3)$ are all known and classified \cite{SenechalAMM1990}. In this subsection, we consider a seed resonator placed at a point in space and with an orientation described by $x_0 \in \bm E$. Copies of this seed resonator are then acted from the left with space transformations contained in a finite subgroup $\Gamma \subset SO(3)$. As such, we generate a finite lattice 
\begin{equation}\label{Eq:LGamma}
    \Ll =\{\gamma x_0, \ \gamma \in \Gamma\}.
    \end{equation}
The seed resonator corresponds to the neutral element $e \in \Gamma$ and is part of lattice $\Ll$.

\begin{remark}\label{Re:FixePoints}{\rm One important assumption of ours is that the cardinals of $\Ll$ and $\Gamma$ always coincide: $|\Ll|=|\Gamma|$. This forbids $x_0$ for being part of the set $\bm E^\Gamma$ of points of $\bm E$ fixed by $\Gamma$.}$\Diamond$
\end{remark}

We want to stress that resulting lattice depends quite sensitively on $x_0$, a fact that was already visible in Fig.~\ref{Fig:IHMolecule}. However, regardless of the initial choice of $x_0 \in \bm E \setminus \bm E^\Gamma$, we have:

\begin{proposition}
    The sets $x \cdot \Ll$ appearing in Eq.~\eqref{Eq:RedDMat} are all identical.
\end{proposition}

\begin{proof} Before giving the arguments, let us point out that $x \cdot \Ll$ for $x \in \Ll$ represents the lattice $\Ll$ as seen from the local reference frame of the resonator located at $x$. Then the statement says that synthetic molecule appears identical when looked at from any of the local reference frames. Now, if $x$ and $y$ belong to $\Ll$, they must be of the form $x=\gamma x_0$ and $y= \gamma' x_0$, for some $\gamma,\gamma' \in \Gamma$. As such, 
\begin{equation}
    x \cdot y = y x^{-1}=  (\gamma' x_0) (\gamma x_0)^{-1}= \gamma \cdot \gamma'.
\end{equation}
Given Eq.~\eqref{Eq:LGamma}, we have
\begin{equation}
    x \cdot \Ll = \gamma \cdot \Gamma = \Gamma,
\end{equation}
where we used the invariance of a group against its own right action.
\end{proof}

The above facts have important consequences: Since the coupling matrices are entirely determined by the local environment as experienced from the local frames of the resonators, we have $w_{x',x}(\Ll)=w_{x \cdot x'}$ and, by setting $x \cdot x' = \gamma$, we have
\begin{equation}\label{Eq:DGamma}
    \bm D = \sum_{\gamma \in \Gamma} w_\gamma \otimes \sum_{x \in \Ll}|\gamma x\rangle \langle x|.
\end{equation}
We will see in the next section that any such $\bm D$ can be generated from the left-regular representations of the (stabilized) group $C^\ast$-algebra of $\Gamma$. As explained in our introductory remarks, this is an effective though somewhat abstract way to communicate the constraints we want to impose on our synthetic molecules. Under these constraints, one is free to modify the position $x_0$ of the seed resonator in $\bm E \setminus \bm E^\Gamma$. One is also free to change the internal structure of the seed resonator, including adding or removing degrees of freedom. For example, we can add a second distinct seed resonator, because the pair can be seen as one seed resonator with a more complex structure. All these mentioned actions  will essentially result in an alteration of the coupling matrices $w_\gamma$ and they lead to a well defined class of synthetic molecules associated with the group $\Gamma$. The common thing about these molecules is that their dynamical matrices are all generated from the same group $C^\ast$-algebra.

Since there is no cap on the number of degrees of freedom carried by the seed resonator, one can create quite complicated molecules with the algorithm just described, with layers upon layers of resonators, each having their own spatial distributions. Due to this untamed complexity, one can be mislead into thinking that there are no limits on the eigenmodes or dynamical patterns one can produce with such molecules. Quite the contrary, once the point group $\Gamma$ is fixed, there are only a finite number of truly distinct dynamical patterns that one can squeeze out of the synthetic molecule (see section~\ref{Sec:FundModels}).

\begin{figure}[t]
\begin{center}
\hspace{10cm}
\includegraphics[width=0.6\textwidth]{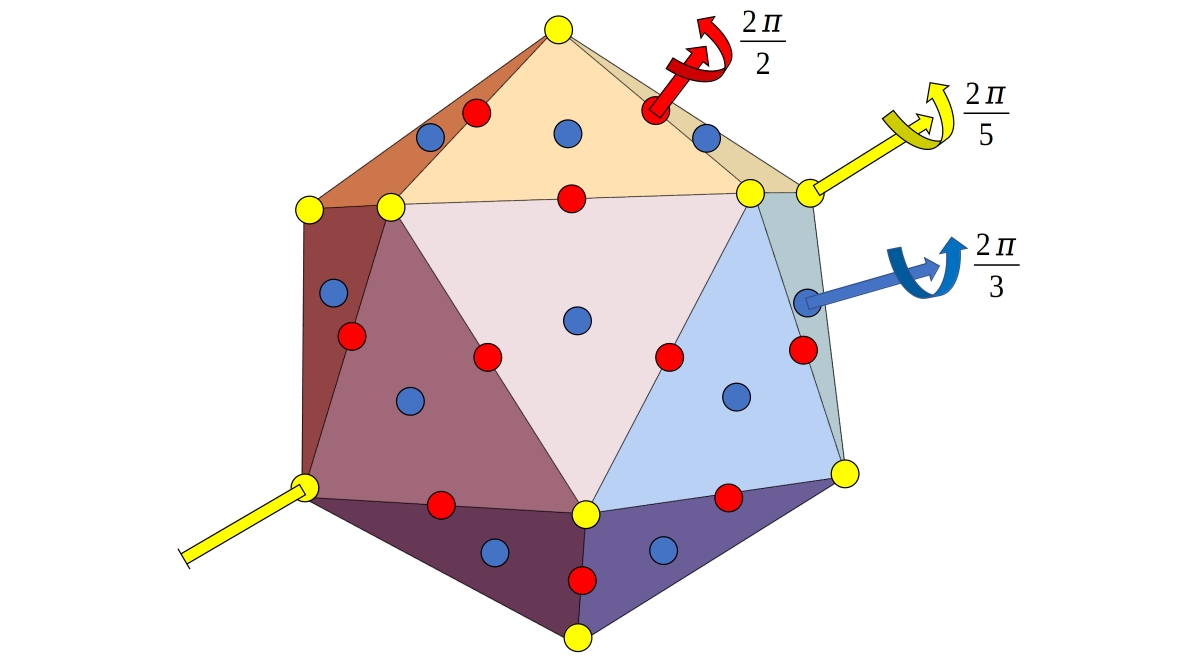}
  \caption{\small The icosahedron and its proper symmetries.
}
 \label{Fig:IcoSymm}
\end{center}
\end{figure}

\begin{example}\label{Re:Ico}{\rm The icosahedral group $I_h$ is the group of symmetries of  the icosahedron, which is one of the existing platonic shapes. The subgroup $I_p$ of $I_h$ containing all proper space transformations has exactly 60 elements. They can be described as follows. There are 5-fold rotations around the six axes passing through pairs of opposing vertices, shown as yellow bullets in Fig.~\ref{Fig:IcoSymm}. One of the mentioned rotation axes is shown in yellow in Fig.~\ref{Fig:IcoSymm}. Since these axes also pass through the center of the icosahedron, we can specify them by using the coordinates $\vec r =(x,y,z)$ of six vertices in a rectangular coordinate system with origin at the center of the icosahedron. Then the axes themselves are encoded in the norm one vectors $\hat n = \vec r /|\vec r|$. The mentioned coordinates are ($\varphi=$ golden ratio, $\epsilon = \pm 1$)
$$
C_5: \quad  (\epsilon , \varphi, 0), (0, \epsilon, \varphi), (\varphi, 0, \epsilon).
$$
Rotations by $\frac{2n\pi}{5}$, $n=\overline{1,4}$, around these six axes supply a total of 24 space transformations. Additionally, there are 3-fold rotations about the ten axes passing through the centers of opposing facets shown as blue bullets in Fig.~\ref{Fig:IcoSymm}, and one of the axes is shown in blue in Fig.~\ref{Fig:IcoSymm}. Again, we specify these axes by giving the coordinates of the centers of ten facets, which are
$$
C_3: \quad \begin{array}{c}
(-\varphi,0,\epsilon(1+2\varphi)), (\epsilon(1+\varphi),1+\varphi,1+\varphi), (0,1+2\varphi,\epsilon \varphi), \\
(-1-2\varphi,\epsilon \varphi,0), (-1-\varphi,\epsilon(1+\varphi),-\epsilon(1+\varphi)).
\end{array}
$$
Rotations by $\frac{2n\pi}{3}$, $n=1,2$, around the ten axes supply a total of 20 space transformations. Lastly, there are 2-fold rotations about the 15 axes passing through the centers of opposing edges, shown as red bullets in Fig.~\ref{Fig:IcoSymm}. One such rotation axis is shown in red in Fig.~\ref{Fig:IcoSymm}. We specify these axes by giving the coordinates of the centers of 15 edges, which are 
$$
C_2: \quad \begin{array}{c}
(\varphi,0,0), (0,\varphi,0),(0,0,\varphi) ,\\
\tfrac{1}{2}(-\epsilon ,\epsilon (1+\varphi), \varphi),\tfrac{1}{2} (\epsilon (\varphi+1), -\epsilon \varphi, 1), \tfrac{1}{2}(\epsilon \varphi, -\epsilon, 1+\varphi),\\
\tfrac{1}{2}(\epsilon ,\epsilon (1+\varphi), \varphi),\tfrac{1}{2} (\epsilon (\varphi+1), \epsilon \varphi, 1), \tfrac{1}{2}(\epsilon \varphi, \epsilon, 1+\varphi).
\end{array}
$$
Rotations by $\pi$ around these 15 axes supply a total of 15 space transformations. Together with the identity transformation, we have listed 60 space transformations, hence, all the elements of the proper icosahedral group.
    }$\Diamond$
\end{example}

\begin{remark}{\rm The procedures used in Example~\ref{Re:Ico} apply to all platonic shapes, in the sense that all their proper point symmetries can be found by identifying the pairs of opposing vertices, opposing edge centers and opposing face centers, together with the appropriate angles of rotations. The latter are easy to determine once the axes of rotations are identified by the pairs, as explained.}$\Diamond$   
\end{remark}

The ensembles seen in Fig.~\ref{Fig:IHMolecule} were produced by applying the transformations listed above on a seeding shape, a cone in this instance. If this seeding shape is fitted with magnets as in Fig.~\ref{Fig:SymCp}, then this procedure produces genuine synthetic molecules from the class associated with the $I_p$ group. On the contrary, the synthetic molecule shown in Fig.~\ref{Fig:RealSynthMolecule} has 12 resonators placed exactly above the vertices of an icosahedron. Thus, this is one of the cases where the seed resonator sits at a forbidden location. Another observation is that the molecule in Fig.~\ref{Fig:IHMolecule}(b) is not as symmetric as the platonic shape, since some pairs of resonators are closer spaced than others. Nevertheless, its lattice is still symmetric under the full $I_p$ group.

\begin{figure}[t]
\begin{center}
\hspace{10cm}
\includegraphics[width=0.5\textwidth]{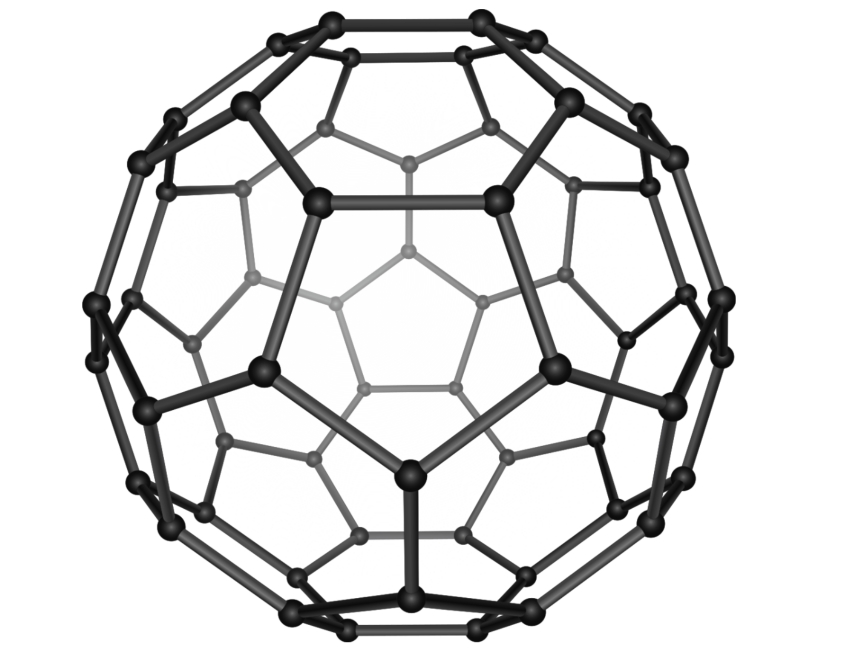}
  \caption{\small The C$_{60}$ molecule contains 60 carbon atoms arranged in a spatial configuration that matches the vertices of a truncated icosahedron. 
}
 \label{Fig:C60}
\end{center}
\end{figure}

Fig.~\ref{Fig:C60} illustrates the C$_{60}$ molecule \cite{KrotoNature1985}, where the 60 carbon atoms sit at the vertices of a truncated icosahedron. The nodes seen in this picture represent carbon atoms. Their outer-shell electrons are $sp^2$-hybridized as in graphene, hence 3 of the valence electrons of a carbon atom are locked in strong covalent chemical bonds, seen as edges in Fig.~\ref{Fig:C60},  while the remaining valence electrons (one per carbon atom) hop between the $\pi$-orbitals carried by each of the atoms. If $|a\rangle$ is the $\pi$-orbital carried by atom $a$ of the molecule, as defined in a local frame attached to the atom, then we can write the discrete Hamiltonian
\begin{equation}\label{Eq:HC60}
H_{C_{60}} = - E_0 \sum_{\langle a,a'\rangle} |a\rangle \langle a'|, \quad E_0 >0,
\end{equation}
where $\langle \, , \, \rangle$ means near neighboring atoms. This Hamiltonian accounts well for the low energy dynamics of the 60 valence electrons of the molecule.

\subsection{The Cayley graph picture} We mentioned at one point the similarity between the expressions~\eqref{Eq:DGamma} and \eqref{Eq:H}, which is a nice feature of the formalism. Still, in Eq.~\eqref{Eq:DGamma}, the lattice needs to be imagined inside an abstract group, but this can be corrected using the notion of Cayley graph.

The isometry class of all space groups can be abstractly presented in terms of generators and relations \cite{CoxeterBook}. In particular, if $\FM_2=\langle X,Y\rangle$ is the free nonabelian group in two generators and $N$ is its normal subgroup generated by the set $\{X^2,Y^2, (XY)^3\}$, then $\FM_2/N$ is isomorphic to $I_p$. Explicitly, the isomorphism is given by the identifications $X \leftrightarrow C_5$ and $Y \leftrightarrow C_2$, where $C_5$ and $C_2$ are the rotations displayed in Fig.~\ref{Fig:IcoSymm}. Thus, $I_p$ can be presented as
\begin{equation}\label{Eq:GenRel1}
I_p = \big \langle C_5, C_2 | C_5^5 , C_2^2 , (C_5 C_2)^3  \big \rangle.
\end{equation}
As such, any element of $I_p$ can be written as a word made up from the two letters $C_5$ and $C_2$, but this writing is not unique due to the presence of relations. One should also be aware that other generating sets and relations are possible, and, to distinguish the presentation of $I_p$ displayed above from others, we will refer to Eq.~\eqref{Eq:GenRel1} as the standard presentation of $I_p$. Such facts apply to any finite subgroup of $SO(3)$.

\begin{definition}[\cite{MeierBook}]\label{Def:CG} Given a discrete group $\Gamma$ and a finite subset $S$ of $\Gamma$, the Cayley graph $\Cc(\Gamma,S)$ is the un-directed graph with vertex set $\Gamma$ and edge set containing an
edge between $\gamma$ and $s \gamma$ whenever $\gamma \in \Gamma$ and $s \in S$.  
\end{definition}

If $S$ is a standard generating set of $\Gamma$, then we call $\Cc(\Gamma,S)$ the standard Cayley graph of $\Gamma$. A more refined and more useful geometric object is the Cayley digraph:

\begin{definition}[\cite{MeierBook}]\label{Def:CDG} Given a discrete group $\Gamma$ and a subset $S$ of $\Gamma$, let $c: S \to {\rm Color}$ assign a distinct color to each $s \in S$. Then the Cayley digraph $\vec \Cc(\Gamma,S,c)$ is the colored graph with vertex set $\Gamma$ and directed edges from $\gamma$ to $s\gamma$ for $\gamma \in \Gamma$ and $s \in S$. All directed edges produced by $s \in S$ are assigned the color $c(s)$. 
\end{definition}

\begin{figure}[t]
\center
\includegraphics[width=0.4\textwidth]{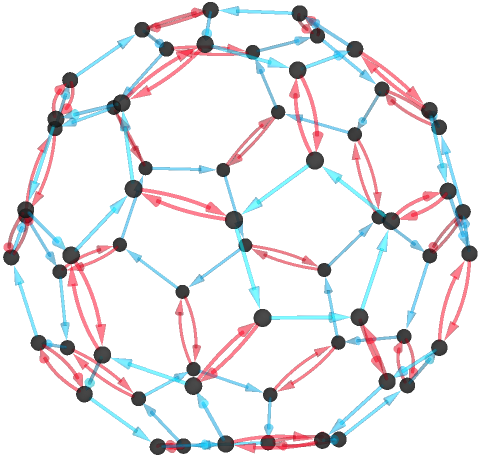}
  \caption{\small Cayley digraph of the proper icosahedral group. Red arrows represent the $C_2$ generator and blue arrows represent the $C_5$ generator.
}
 \label{Fig:IcoDiag}
\end{figure}

The standard Cayley digraph of $I_p$ is illustrated in Fig.~\ref{Fig:IcoDiag}, which renders the group in a geometric fashion. We can select any vertex of this graph to represent the neutral element $e$, and various paths between $e$ and a vertex $\gamma \in \Gamma$ supply a valid word for $\gamma$ made up from the letters $C_5$ and $C_2$ and their inverses. The Hamiltonian~\eqref{Eq:HC60} coincides, up to a factor, with the adjacency operator of this graph
\begin{equation}\label{Eq:Adjacency}
\Delta = - \sum_{\gamma \in I_p} \big (\tfrac{1}{2}|C_5 \gamma \rangle \langle \gamma| +\tfrac{1}{2} |C_5^{-1} \gamma\rangle \langle \gamma| + |C_2 \gamma\rangle \langle \gamma| \big ).
\end{equation}
We will see in subsection~\ref{Sec:SpecEng} how such Hamiltonians can be graphically represented on the Cayley graph.

As advertised in \cite{LuxAHP2024}, Cayley graphs can be used to produce real synthetic molecules from abstract finitely generated groups. Furthermore, Cayley graphs can be used very effectively to render the resonant modes of molecules.

\section{The $C^\ast$-algebra of dynamical matrices}\label{Sec:GCAlg}

This section establishes the connection between the abstract $C^\ast$-algebra of a point group and the algebra of dynamical matrices for the class of synthetic molecules introduced in the previous section.

\subsection{Group $C^\ast$-algebras}\label{Sec:CG} The material of this subsection is standard and it can be found, for example, in \cite{DavidsonBook}[Ch.~VIII]. Given a generic discrete group $\Gamma$, its group algebra $\CM \Gamma$ consists of formal series
\begin{equation}\label{Eq:CG}
q = \sum_{\gamma \in \Gamma} \alpha_\gamma \, \gamma, \quad \alpha_\gamma \in \CM,
\end{equation}
where all but a finite number of terms are zero.\footnote{Of course, this is of no concern if $\Gamma$ is finite.}
The addition and multiplication of such formal series are defined in the obvious way, using the group and algebraic structures of $\Gamma$ and $\CM$, respectively. 
In addition, there exists a natural $\ast$-operation
\begin{equation}
q^\ast = \sum_{\gamma \in \Gamma} \alpha_\gamma^\ast \, \gamma^{-1}, \quad (q^\ast)^\ast = q, \quad (\alpha q)^\ast = \alpha^\ast q^\ast, \ \alpha \in \CM.
\end{equation}
Hence, $\CM \Gamma$ is a $\ast$-algebra in a natural way.

\vspace{0.1cm}

We denote by $e$ the neutral element of $\Gamma$. Then the map
\begin{equation}
\Tt : \CM \Gamma \to \CM, \quad \Tt(q) = \alpha_e
\end{equation}
defines a faithful positive trace on $\CM \Gamma$ and a pre-Hilbert structure on $\CM \Gamma$ via
\begin{equation}
\langle q, q' \rangle : = \Tt(q^\ast q'), \quad q,q' \in \CM \Gamma.
\end{equation}
The completion of the linear space $\CM G$ under this pre-Hilbert structure supplies the Hilbert space $\ell^2(\Gamma)$, spanned by $|\gamma\rangle$, $\gamma \in \Gamma$, which form an orthonormal basis
\begin{equation}
\langle \gamma , \gamma' \rangle = \delta_{\gamma,\gamma'}, \quad \gamma,\gamma' \in \Gamma.
\end{equation} 
The left action of $\CM \Gamma$ on itself can be extended to the action of a bounded operator on $\ell^2(\Gamma)$, and this supplies the left regular representation $\pi_L$ of $\CM \Gamma$ inside the algebra $\BM\big(\ell^2(\Gamma)\big)$ of bounded operators on $\ell^2(\Gamma)$. Specifically,
\begin{equation}\label{Eq:PiLeftReg}
\pi_L(q) |\gamma'\rangle = \sum_{\gamma\in \Gamma} \alpha_{\gamma} |\gamma \gamma' \rangle, \quad q \in \CM \Gamma, \quad \gamma' \in \Gamma,
\end{equation} 
where $q = \sum_{\gamma \in \Gamma} \alpha_{\gamma}  \gamma$ with $\alpha_{\gamma} \in \CM$.
The completion of $\CM \Gamma$ with respect to the norm
\begin{equation}\label{Eq:CNorm}
\|q\| : = \|\pi_L(q)\|_{\BM(\ell^2(\Gamma))}
\end{equation}
supplies the reduced group $C^\ast$-algebra $C^\ast_r \Gamma$ of $\Gamma$.

\begin{remark}{\rm It may seems that, for finite groups, we can equally well work with the group algebra. However, the norm~\eqref{Eq:CNorm} puts a topology on the algebra, which will then enable us to speak of ``continuous deformations'' of the models. 
}$\Diamond$
\end{remark}

The right-regular representation of the group algebra acts on $\ell^2(\Gamma)$ as
\begin{equation}\label{Eq:PiRightReg}
\pi_R(q) |\gamma'\rangle : = \sum_{\gamma \in \Gamma} \alpha_{\gamma} |\gamma' \gamma^{-1} \rangle, \quad q \in C^\ast_r \Gamma, \quad \gamma' \in \Gamma.
\end{equation} 
\begin{remark}\label{Re:Tau}{\rm We note that, if the group $\Gamma$ is finite, then
\begin{equation}
\Tt(q) = \frac{1}{|\Gamma|} {\rm Tr}\big ( \pi_L(q) \big ) = \frac{1}{|\Gamma|} {\rm Tr}\big ( \pi_R(q) \big ),
\end{equation}
which can be seen straight from \eqref{Eq:PiLeftReg} and \eqref{Eq:PiRightReg}.}$\Diamond$
\end{remark}

\subsection{Symmetry and structure of dynamical matrices} The right-regular representation induces a unitary representation of $\Gamma$ on $\ell^2(\Gamma)$:
\begin{equation}
U : \Gamma \to \BM\big ( \ell^2(\Gamma) \big ), \quad U_\gamma |\gamma' \rangle : = \pi_R(\gamma)|\gamma'\rangle = |\gamma' \gamma^{-1}\rangle.
\end{equation}
Since the left and right actions commute, we automatically have that
\begin{equation}
U_\gamma^\ast \, \pi_L(q) \, U_\gamma = \pi_L(q), \quad q \in C^\ast_r \Gamma, \quad \gamma \in \Gamma.
\end{equation}
Hence, the group $C^\ast$-algebra supplies models that are naturally $\Gamma$-symmetric. 

Now, given the action~\eqref{Eq:PiLeftReg} of left-regular representation on the basis of $\ell^2(\Gamma)$, we can write
\begin{equation}\label{Eq:PiLQ}
    \pi_L(q) = \sum_{\gamma \in \Gamma} \alpha_\gamma \sum_{\gamma' \in \Gamma}|\gamma \gamma'\rangle \langle \gamma'|.
\end{equation}
Using the isomorphism of Hilbert spaces $\ell^2(\Ll) \simeq \ell^2(\Gamma)$, $|x\rangle = |\gamma x_0\rangle \mapsto |\gamma\rangle$, we can compare \eqref{Eq:PiLQ} with \eqref{Eq:DGamma} to conclude the $\bm D$ is just the left-regular representation of the element $h=\sum_\gamma w_\gamma \gamma \in C^\ast_r \Gamma$. The only difference is that $w_\gamma$ is a matrix in general and we can correct for it by tensoring $C^\ast_r \Gamma$ with $\KM$, the algebra of compact operators over $\ell^2(\Dd)$, where $\Dd$ is the set of internal degrees of freedom.\footnote{To allow unrestricted turning off and on of the degrees of freedom, $\Dd$ is assumed to be infinite.} This operation is commonly referred to as the stabilization of the algebra. At this point, we reached one of our main points:

\begin{proposition} All dynamical matrices for the class of synthetic molecules associated with the group $\Gamma$ can be generated as left-regular representations of self-adjoint elements from $\KM \otimes C^\ast_r \Gamma$.
\end{proposition}

\begin{example}{\rm The adjacency operator~\eqref{Eq:Adjacency} is the left-regular representation of the self-adjoint element $\delta = \frac{1}{2}(C_5+C_5^{-1}) + C_2$ from $C^\ast_r I_p$.}$\Diamond$
\end{example}

\section{Algebraic and Topological Aspects of the Dynamics}
\label{Sec:AlgTopDyn}

The subject of this section is on the decomposition of the linear space of resonant modes in the smallest possible subspaces that remain invariant under dynamics, as well as on the classification of these subspaces from algebraic and topological point of views. The former one engages the ordinary representation theory while the latter one requires more specialized tools, specifically, Kasparov's bi-variant K-theory. While an in-depth mastery of the latter requires substantial effort, the concepts and the working principles can be communicated in very few words. For a reader unfamiliar with the theory, seeing them in action in the simplest but non-trivial context of finite groups and algebras can serve as a door to the world of these ideas. This is one of the main reasons KK-theory appears in our exposition. Another reason is that there is no other natural alternative to the topological point of view. In our opinion, KK-theory is a must-have tool for anyone planning to enter the business of topological dynamics. We hope that the patient reader can sense that from the exposition below.

\subsection{The algebraic representation ring}
\label{Sec:AlgRepR}

A unitary representation of a group $\Gamma$ is a group morphism between the group and the group $\UM(\Hh)$ of unitary operators over a Hilbert space $\Hh$. Thus, a unitary representation $\pi : \Gamma \to \UM(\Hh)$, typically shorthanded to just $\pi$, specifies both the Hilbert space $\Hh$ and the morphism itself. If the Hilbert space has a finite dimension $N$, then we are talking about finite dimensional representations, which are nothing but group morphisms from $\Gamma$ to $U(N)$, the group of unitary matrices of rank $N$. Two representations $\pi_1: \Gamma \to \UM(\Hh_1)$ and $\pi_2: G \to \UM(\Hh_2)$ are said to be unitarily equivalent if there exists a unitary map $U: \Hh_1 \to \Hh_2$ such that $\pi_1 = U^\ast \pi_2 U$. Representation theory is the art of classifying the representations of a group up to unitary equivalence. Since the relation between finite groups and their group algebras is functorial, any finite group representation of a finite group extends by linearity to a finite representation of its algebra as bounded linear operators over the same Hilbert space. Reciprocally, any finite representation of a group's algebra supplies, by restriction to the group elements, a representation of the group.

Given two finite group representations $\pi_1$ and $\pi_2$, we can define a new finite representation by using their direct sum $\pi_1 \oplus \pi_2$. This binary operation is compatible with the unitary equivalence, hence it descends on the set of unitarily equivalent classes of finite representations, where it becomes an abelian binary operation. Thus, the set $\Rr$ of classes of unitarily equivalent finite group representations has a natural structure of an abelian semigroup. It can be closed to a group $(\Rr,\oplus)$, which is called the group of virtual representations. 

\begin{table}[t!]
\begin{center} 
\begin{tabular}{c|c c c c c}
\hline
$ \mathbf{I_p} $ & $E$ & $12C_5$ & $12C_5^2$ & $20C_3$ & $15C_2$ \\ \hline

$A_g$ & 1 & 1 & 1 & 1 & 1 \\ 
$T_{1g}$ & 3 & $\varphi$ & $1-\varphi$ & 0 & -1 \\ 

$T_{2g}$ & 3 & $1-\varphi$ & $\varphi$ & 0 & -1 \\ 

$G_g$ & 4 & -1 & -1 & 1 & 0  \\ 

$H_g$ & 5 & 0 & 0 & -1 & 1  \\ 
\end{tabular}
\end{center}
\caption{The character table of $I_p$ group, where $\varphi$ is the golden ratio.}
\label{Tb:1}
\end{table}

A representation $\pi: \Gamma \to \UM(\Hh)$ is said to be irreducible if there is no linear subspace of $\Hh$ that is left invariant by the actions of all $\pi(\gamma)$, $\gamma\in \Gamma$. Up to unitary equivalences, the irreducible representations of a finite group $\Gamma$ can be enumerated using the characters of the group, which are readily available in the literature. For example, the table of characters of the proper icosahedral group can be found in \cite{LiboffBook} and is reproduced in Table~\ref{Tb:1}. There is a simple and practical way to generate an irreducible representation $\pi_\chi$ from a character $\chi$ (see section~\ref{Sec:GenFM}). Any finite group representation $\pi$ is unitarily equivalent to a finite direct sum of irreducible representations
\begin{equation}
\pi \simeq \bigoplus\nolimits_\chi (\pi_\chi \oplus \cdots \oplus \pi_\chi)
\end{equation} 
The above can be written more compactly as $\pi \simeq \oplus_\chi \, n_\chi \, \pi_\chi$, $n_\chi \in \NM$, which shows that $(\Rr,\oplus)$ is freely and commutatively generated over the classes of irreducible representations.

There is additional algebraic structure on $\Rr$. Indeed, given two finite group representations $\pi_1$ and $\pi_2$, we can generate a new finite representation using their tensor product $\pi_1 \otimes \pi_2$. This binary operation is also compatible with the unitary equivalence, hence it descends on $\Rr$, where it becomes an associative commutative binary operation, which distributes over $\oplus$. Thus, the set $\Rr$ of classes of unitarily equivalent finite group virtual representations has a natural structure of a commutative ring, $(\Rr,\oplus,\otimes)$, called the ring of virtual representations \cite{Serre1977LinearRO}.

\subsection{Dynamics vs group representations}
\label{Sec:DynVsRep}

When analyzing the dynamics of a synthetic molecule, the ring of representations comes into play in the following way. Let $\bm D$ be a dynamical matrix as in \eqref{Eq:DGamma}. As we have already seen, $\bm D$ is a linear operator over the Hilbert space $\Hh=\CM^D \otimes \ell^2(\Gamma)$, which is necessarily the left-regular representation of an element $h$ from the group algebra $\CM \Gamma$. Suppose we compute or measure the resonant spectrum of $\bm D$ and we find an eigenvalue $\lambda$. By standard methods, we can associate a spectral projection $P_\lambda$ such that $P_\lambda \Hh$ is the linear space spanned by all the eigenvectors of $\bm D$ corresponding to $\lambda$. One finds that this $P_\lambda$ is also the left-regular representation of a projection $p_\lambda$ from $\CM \Gamma$. For example, consider the so-called interpolating polynomials defined by the relations $F_\lambda(\lambda') = \delta_{\lambda,\lambda'}$ for any $\lambda$ and $\lambda'$ from the discrete spectrum of $\bm D = \pi_L(h)$. Recalling that the finite representations commute with polynomial calculus, we have  
\begin{equation}
P_{\lambda} = F_{\lambda}(\bm D) = F_{\lambda}(\pi_L(h))=\pi_L(F_{\lambda}(h)),
\end{equation}
which shows that $P_{\lambda}$ is the left-regular representation of $p_{\lambda}=F_{\lambda}(h) \in \CM \Gamma$. Lastly, since the left and right regular representations commute, we see that
\begin{equation}\label{Eq:PiLambda}
\pi_\lambda: \Gamma \to \BM(P_\lambda \Hh), \quad \pi_\lambda (\gamma) : = P_\lambda \pi_R(\gamma) P_\lambda,
\end{equation}
defines a finite representation of $\Gamma$. We learned in the previous sub-section that each $\pi_\lambda$ is unitarily equivalent to a direct sum of irreducible representations. Thus, we have the following inescapable fact:

\begin{proposition} Up to a unitary transformation, all resonant modes, without exception, coincide with a vector from the representation space of one of the $\pi_\chi$'s, or linear combinations of such vectors.
\end{proposition} 

The above fact shows that the dynamical patterns supported by a synthetic molecule from the class associated to the group $\Gamma$ can be all generated as linear combinations of vectors drawn from a finite number of spaces, which can be enumerated by the character of the group. The statement is purely algebraic and, as such, it cannot reveal anything about the stability of those dynamical patterns against continuous deformations of the molecule. In the light of this remark, it is then natural to approach the representation theory and the classification of dynamical patterns from a topological point of view. This is the subject of the following two subsections.

\subsection{The topological representation ring} 
\label{Sec:TopRepR}

We describe here the so called ring of Fredholm representations introduced by Kasparov \cite{KasparovConspectus}. This is achieved through Kasparov's bi-variant K-theory \cite{KasparovJSM1981,KasparovJSM1987}, also known as KK-theory, which has many equivalent presentations. In Cuntz's picture \cite{CuntzKTheory1987}, for two $C^\ast$-algebras $A$ and $B$, an element of Kasparov's abelian group $KK_0(A,B)$ is presented as the homotopy class of a pair of morphisms $\bar \phi =(\phi_1,\phi_2)$ from $A$ to the algebra of adjointable endomorphisms over the standard Hilbert $C^\ast$-module $\ell^2(\NM,B)$, with the property that the difference $\phi_1 -\phi_2$ takes values in $\KM \otimes B$. In particular, any morphism $\psi : A \to \KM \otimes B$ supplies the element $[(\psi,0)]$ of $KK_0(A,B)$, where $[\cdot]$ indicates the homotopy class. The abelian group structure on $KK_0(A,B)$, denoted by ``+'' in the following, is induced by the direct sum of algebra morphisms.

A powerful feature of Kasparov's K-theory is the existence of multiplicative structures. If $A$, $B$, $C$, $D$ are $C^\ast$-algebras, then there is an associative internal product
\begin{equation}
KK_0(A,B) \times KK_0(B,C) \to KK_0(A,C)
\end{equation}
and an associative external product
\begin{equation}
KK_0(A,B) \otimes KK_0(C,D) \to KK_0(A\otimes C, B \otimes D).
\end{equation}
In particular, if $\psi : A \to B$ and $\psi' : B \to C$ are $C^\ast$-algebra morphisms, then
\begin{equation}
[\psi ] \times [\psi'] = [\psi' \circ \psi],
\end{equation}
and if $\psi : A \to C$ and $\psi' : B \to D$ are $C^\ast$-algebra morphisms, then 
\begin{equation}
[\psi ] \otimes [\psi'] = [\psi \otimes \psi'].
\end{equation}
See \cite{KasparovJSM1981,KasparovJSM1987} or the textbook \cite{BlackadarBook} for more details.

\begin{remark}{\rm A formal algorithm for the internal product can be formulated in Cuntz's picture of KK-theory, but it cannot be explained here because it requires knowledge of additional structures.  However, for the present context, the special cases mentioned above suffice (see subsection~\ref{Sec:Inv}). A curious reader can find in \cite{ProdanRMP2016} explicit algorithms for other situations that often occur in physical applications. }$\Diamond$
\end{remark}  

To introduce the ring of Fredholm representations of a group $\Gamma$, we focus on $KK_0(C^\ast_r \Gamma, \CM)$. Note that the elements of $KK_0(C^\ast_r \Gamma, \CM)$ are pairs $(\phi_1,\phi_2)$ of algebra morphisms from $C^\ast_r \Gamma$ to $\BM(\Hh)$, the algebra of bounded operators over a separable Hilbert space, such that $\phi_1 -\phi_2$ lands in $\KM(\Hh)$, the subalgebra of compact operators. This difference is erased if we take the quotient of $\BM(\Hh)$ by $\KM(\Hh)$ and, as such, the two morphisms $\phi_1$ and $\phi_2$ are identical as morphisms into the Calkin algebra $\QM(\Hh) \simeq \BM(\Hh)/\KM(\Hh)$. Restricting to the group elements, this provides a representation of the group $\Gamma$ inside the group of invertible elements of $\QM(\Hh)$, that is, in the group of Fredholm operators. Hence, the name of Fredholm representations. The ring of Fredholm representations exists for general locally compact groups but its construction in this general setting requires equivariant KK-theory. The latter is more technical and harder to explain with few words than the ordinary KK-theory outlined above. For this reason, we follow a construction due to Cuntz \cite{CuntzJRAM1983}, which works for discrete groups, finite or not. In this approach, one uses both the internal and external products to define the ring structure. The key is to see that the group morphism $\Gamma \ni \gamma \mapsto (\gamma,\gamma) \in \Gamma \times \Gamma$ can be lifted to a morphism of $C^\ast$-algebras\footnote{This morphism supplies the standard co-product on $C^\ast_r \Gamma$.}
\begin{equation}
C^\ast_r \Gamma \ni q=\sum_\gamma \alpha_\gamma \, \gamma \mapsto \eta(q) = \sum_\gamma \alpha_\gamma \, \gamma \otimes \gamma \in C^\ast_r \Gamma \otimes C^\ast_r \Gamma \simeq C^\ast_r(\Gamma \times \Gamma).
\end{equation} 
Furthermore, the homotopy class $[\eta]$ supplies an element of $KK_0(C^\ast_r \Gamma, C^\ast_r \Gamma \otimes C^\ast_r \Gamma)$. Then, given two Fredholm representations, hence two members $[\phi]$ and $[\phi']$ of $KK_0(C^\ast_r \Gamma, \CM)$, we can first take their external product $[\phi] \otimes [\phi']$, to land in $KK_0(C^\ast_r \Gamma \otimes C^\ast_r \Gamma, \CM)$, and then take the internal product of the result with $[\eta]$ on the left, to land back in  $KK_0(C^\ast_r \Gamma, \CM)$. In this way, one defines the product
\begin{equation}
 KK_0(C^\ast_r \Gamma, \CM) \ni [\phi], [\phi'] \mapsto  [\phi] \cdot [\phi'] : = [\eta] \times ([\phi] \otimes [\phi']) \in KK_0(C^\ast_r \Gamma, \CM),
\end{equation}
which is a commutative binary operation that distributes over abelian group structure of $KK_0(C^\ast_r \Gamma, \CM)$ \cite{CuntzJRAM1983}. 

The resulting ring structure $(KK_0(C^\ast_r \Gamma, \CM),+,\cdot)$ is the ring of Fredholm representations. It supplies the means to investigate the representation theory of a group, this time using homotopy instead of unitary equivalence. We have the following important fact:

\begin{proposition}[\cite{KasparovConspectus}, p.~111] For a finite group, the ring of Fredholm representations and the ordinary ring of representations coincide.
\end{proposition} 

As such, $KK_0(C^\ast_r \Gamma, \CM)$ is generated by the homotopy classes $[\pi_\chi:C^\ast_r \Gamma \to \KM]$ of the irreducible representations of the group, as lifted to $C^\ast \Gamma$. The main point here is that we replaced the unitary equivalence with the homotopy equivalence of representations. This, for example, gives us automatic assurance that the representation~\eqref{Eq:PiLambda} stays in the same class as long as the projection $p_\lambda$ varies continuously. This is the case if the dynamical matrix is deformed continuously and the eigenvalue $\lambda$ stays isolated. Furthermore, the homotopy classes of the representations act as topological charges for the energy levels. Indeed, when two or more energy levels collide and split again during a deformation, we can be sure that the net topological charge carried by the levels is conserved. Thus, the topological charges are never lost but are redistributed during collisions of energy levels.

\subsection{The K-theoretic group} K-theory is the art of classifying projections with respect to stable homotopy. Textbooks offering introductions to the subject are the monograph by Efton Park \cite{ParkBook}, for topological spaces, and the one by Rordam and Larsen \cite{RordamBook}, for $C^\ast$-algebras. We recommend they be read in the order we mentioned them.

A projection from a generic $C^\ast$-algebra is an element $p$ with the properties $p^2=p^\ast=p$. Given the group $\Gamma$, we have seen that the $C^\ast$-algebra covering all dynamical matrices of molecules from the class associated to $\Gamma$ is $\KM \otimes C^\ast_r \Gamma$. We will organize the projections from this algebra by the homotopy equivalence which says that $p \sim_h p'$ iff there exists a continuous family $p(t)$ of projections in $\KM \otimes C^\ast_r \Gamma$ such that $p(0)=p$ and $p(1)=p'$. Note that along such homotopy, the number of engaged degrees of freedom can vary. Thus, the stated equivalence relation is stronger than the homotopy equivalence defined for a fixed number of degrees of freedom, {\it i.e.} with deformations occurring in $M_{D \times D}(\CM) \otimes C^\ast_r \Gamma$ with a fixed $D$. The terminology of stable homotopy is used to convey this important fact and the equivalence class of a projection is usually denoted by $[p]_0$. To define an additive structure on the set $K_0(C^\ast_r \Gamma)$ of these classes, one observes that any projection from $\KM \otimes C^\ast G$ actually comes from the subalgebra $M_{N \times N}(\CM) \otimes C^\ast_r \Gamma$, for some finite $N$ (see {\it e.g.} \cite{ProdanRMP2016}[Prop.~4.5]). If $p \in M_{N\times N}(\CM) \otimes C^\ast_r \Gamma$ and $p' \in M_{K\times K}(\CM) \otimes C^\ast_r \Gamma$ are two such projections, then 
\begin{equation}
    \begin{pmatrix} p & 0 \\ 0 & p' \end{pmatrix} \in M_{(N+K)\times (N+K)}(\CM) \otimes C^\ast_r \Gamma \subset \KM \otimes C^\ast_r \Gamma
\end{equation} 
and 
\begin{equation}
[p]_0 + [p']_0 := \left [ \begin{pmatrix} p & 0 \\ 0 & p' \end{pmatrix}\right]_0
\end{equation}
endows $K_0(C^\ast_r \Gamma)$ with a commutative binary operation. The abelian group $(K_0(C^\ast_r \Gamma),+)$ encodes the K$_0$-theory of the group algebra.

K$_0$-theory matches perfectly the following classification principle: two energy levels of two distinct dynamical matrices are topologically equivalent if we can deform the configurations of the synthetic molecules until the spectral projections of the levels coincide, while keeping the energy levels spectrally isolated. It is important to add that the deformations should occur inside the specified class of molecules and that degrees of freedom can be turned on and off during this process. Almost tautologically, the stable homotopy class of a projection, that is $[p]_0$, represents the complete topological invariant associated to $p \in \KM \otimes C^\ast_r \Gamma$.

\subsection{Numerical invariants}
\label{Sec:Inv}

Any projection $p \in \KM \otimes C^\ast_r \Gamma$ sets an algebra morphism from $\CM$ to $\KM \otimes C^\ast_r \Gamma$. Indeed, if we define $\CM \ni c \mapsto c p \in \KM \otimes C^\ast_r \Gamma$, then $c_1 c_2 \mapsto (c_1 c_2) p = (c_1 p) (c_2 p)$, where we used the property $p^2=p$ of a projection. We denote this morphism by the symbol $\tilde p$. According to our short description of $KK$-theory, this morphism defines an element $[\tilde p] \in KK_0(\CM, C^\ast_r \Gamma)$. Additionally, we already seen that a finite dimensional representation $\pi$ of $\Gamma$ defines an element of $KK_0(C^\ast_r \Gamma,\CM)$. The two mentioned elements can be paired via Kasparov internal products of the type
\begin{equation}\label{Eq:KKProd1}
KK_0(\CM, C^\ast_r \Gamma) \times KK_0(C^\ast_r \Gamma,\CM) \to KK_0(\CM,\CM) \simeq \ZM,
\end{equation}
to produce numbers that depend only on the homotopy classes of $p$ and $\pi$, hence topological invariants.

\begin{proposition}
\label{prop:K0labeling}
Let $p$ be a projection from $M_{N\times N}(\CM) \otimes C^\ast_r \Gamma \subset \KM \otimes C^\ast_r \Gamma$ and $\{\pi_\chi\}$ be a complete set of irreducible representations of $\Gamma$. Recall that $\{\pi_\chi\}$ extend to morphisms from $\KM \otimes C^\ast_r \Gamma$ to $\KM$, whose homotopy classes generate a basis for Kasparov's group $KK_0\big (C^\ast_r \Gamma, \CM\big )$. Then the pairings $[\tilde p] \times [\pi_\chi]$ are given by
\begin{equation}\label{Eq:Pairing2}
n_\chi : = [\tilde p] \times [\pi_\chi] = \frac{1}{|\Gamma|}\sum_{\gamma \in \Gamma}  {\rm Tr}\big (U_{\gamma} \pi_L( p)\big ) \, \chi(\gamma) \in \ZM,
\end{equation}
and they supply a maximal set of independent numerical invariants for $p$. Above, ``$\times$'' refers to the Kasparov product~\eqref{Eq:KKProd1} and ${\rm  Tr}$ is the trace over $\KM \otimes \BM(\ell^2(\Gamma))$.
\end{proposition} 

\begin{proof} Note that $({\rm id} \otimes \pi_\chi) \circ \tilde p$ is the morphism from $\CM$ to $\KM$ generated by the projection $({\rm id} \otimes \pi_\chi)(p) \in \KM$. We claim that the right side of Eq.~\eqref{Eq:Pairing2} is nothing else but the trace of this projection. According to our observation at the beginning of subsection~\ref{Sec:TopRepR}, $({\rm id} \otimes \pi_\chi) \circ \tilde p$ encodes the Kasparov product $[\tilde p] \times [\pi_\chi]$, which lands in $KK_0(\CM,\CM)$, and the trace applied on $({\rm id} \otimes \pi_\chi)(p)$ is simply the isomorphism mapping $KK_0(\CM,\CM)$ to $\ZM$. Thus, the proof is complete if we can confirm the above claim. Now, suppose first that $p$ is from $C^\ast_r \Gamma$. Then, if $p = \sum_\gamma p_\gamma \, \gamma$, then the coefficients of this expansion are given by (see Remark~\ref{Re:Tau})
\begin{equation}
p_\gamma = \Tt (p \gamma^{-1}) = \frac{1}{|\Gamma|}{\rm Tr} \big (U_\gamma \, \pi_L(p) \big ),
\end{equation}
where the trace is over $\BM(\ell^2(\Gamma))$. More generally, if $p$ is from $M_{N\times N}(\CM) \otimes C^\ast_r \Gamma$, then the coefficients $p_\gamma$ are from $M_{N\times N}(\CM)$ and $p_\gamma = ({\rm id} \otimes \Tt) (p \gamma^{-1})$. Consequently, 
\begin{equation}
{\rm Tr}(p_\gamma) = \frac{1}{|\Gamma|}{\rm Tr} \big (U_\gamma \, \pi_L(p) \big ),
\end{equation}
where the trace on the right is over $\KM \otimes \BM(\ell^2(\Gamma))$. Then
\begin{equation}
{\rm Tr}\Big (({\rm id} \otimes \pi_\chi)(p) \Big ) = {\rm Tr}\Big (\sum_{\gamma \in \Gamma} p_\gamma \otimes \pi_\chi(\gamma) \Big ) = \frac{1}{|\Gamma|}\sum_{\gamma \in \Gamma} {\rm Tr}(U_\gamma \pi_L(p)\big ) {\rm Tr}\big (\pi_\chi(\gamma) \big ),
\end{equation}
and this completes the proof.
\end{proof}

\section{Engineering the Fundamental Models}
\label{Sec:FundModels}

Consider a symmetric Hamiltonian $H$ on $\CM^N \otimes \ell^2(\Gamma)$ with a gap in its spectrum and denote by $P_G$ the spectral projection onto the spectrum below this gap. This projection is often called the gap projection. It defines a projection $p_G \in \KM \otimes C^\ast \Gamma$ and, as we have seen in the previous section, $p_G$ can be stably deformed into a direct sum of projections $\oplus_\chi \oplus_{i=1}^{n_\chi} p_\chi$ that runs over the generating set of the $K_0$-group. At its turn, this assures us that the original gapped Hamiltonian can be stably deformed into $-\oplus_\chi \oplus_{i=1}^{n_\chi} \pi_L(p_\chi)$ without closing the gap. Thus, if we define the Hamiltonians $H_\chi = - \pi_L(p_\chi)$, then, up to continuous deformation, $H$ is equivalent to a direct sum of $H_\chi$'s. In other words, any class of gapped Hamiltonians can represented by  a stack of disconnected models $H_\chi$. Hence, the latter can be rightfully called the fundamental models in the context of $\Gamma$-symmetric models. The goal of this section is to supply techniques that generate explicit expressions for $H_\chi$'s that involve a small number of entries. 

\subsection{Generating the fundamental models}\label{Sec:GenFM}  Our first task is to supply an explicit realization of a generating set $\{p_\chi\}$ of projections for the $K_0$-group. Consider a generic self-adjoint element from $\CM \Gamma$, $h = \sum_{g \in \Gamma} \alpha_\gamma \cdot \gamma$. Since the sum involves $|\Gamma|$ complex coefficients, we will view $h$ as living in the $\CM^{|\Gamma|}$ parameter space. The left-regular representation $\pi_L(h)$ on the Cayley graph of $\Gamma$ has a pure-point spectrum and the eigen-spaces corresponding to the eigenvalues supply irreducible representations for $\Gamma$, unless the coefficients $\alpha_\gamma$ enter in specific relations. In other words, the eigen-spaces $h$ supply irreducible representations for $\Gamma$ except when the set of coefficients $\{\alpha_\gamma\}$ belongs to a sub-manifold of the full parameter space of strictly smaller dimension. As such, if we generate the coefficients $\alpha_\gamma$ from a ball of $\CM^{|\Gamma|}$ via a random process, then, with probability one, the eigen-spaces of $h$ will supply irreducible representations of $\Gamma$. Since each of them shows up in the regular representation of the group, this technique gives us a practical way to sample all irreducible representations of $\Gamma$. If $\lambda_\chi$ and its spectral projector 
\begin{equation}
P_\chi = \sum_{\pi_L(h)\ket{\eta^{(\chi)}_i} = \lambda_\chi \ket{\eta^{(\chi)}_i}} \ket{\eta^{(\chi)}_i}\bra{\eta^{(\chi)}_i} 
\end{equation}
were identified to relate to a representation $\chi$, then the sought projection can be computed as
\begin{equation}\label{Eq:PChi2}
p_\chi = \sum_{\gamma \in \Gamma} \langle e | P_\chi | \gamma \rangle \, \gamma.
\end{equation}
For the proper icosahedral group $I_p$, we don't need to look any further than the adjacency operator $\Delta$ because each irreducible representation of $I_p$  appears as one of its eigen-spaces. While our goal for the section is theoretically completed, the solution we have now is not entirely satisfactory from the practical point of view. Indeed, if we want to implement $H_\chi = - \pi_R(p_\chi)$ experimentally, we will find it quite difficult because each term appearing in the sum~\eqref{Eq:PChi2} requires a physical coupling. Thus, to reach a practical level, we must reduce the number of terms in Eq.~\eqref{Eq:PChi2} as much as possible, without changing the topological class of $H_\chi$. For this purpose, we present two methods. 

The first method rests on the simple observation that the self-adjoint operator $(h-\lambda_\chi)^2$ has a positive spectrum which contains $0$ and the spectral projector corresponding to the $0$ eigenvalue coincides with $P_\chi$. Then we can choose $h_\chi = (h-\lambda_\chi)^2$ if we manage to generate a $h \in \CM G$ with a short coupling range. With the proposed random process introduced above, this can be achieved by sampling the elements with coupling range 1, moving to elements with coupling range 2 if not successful, and so on. In the case of the proper icosahedral group, the adjacency operator~\eqref{Eq:Adjacency} has coupling range 1, hence it is optimal and we can choose $h_\chi = (\Delta - \lambda_\chi)^2$, which has coupling range 2. Its corresponding algebra element has the following simple expression:
\begin{equation}
\label{eqn:toymodel}
\begin{aligned}
h_\chi & = (\tfrac{1}{2} C_5 + \tfrac{1}{2} C_5^{-1} +C_2 +\lambda_\chi)^2 \\
& =  (\tfrac{3}{2} + \lambda^2_\chi) e +\lambda_\chi (\tfrac{1}{2} C_5 + \tfrac{1}{2} C_5^{-1} +C_2) \\
& \quad + \tfrac{1}{4} C_5^2 + \tfrac{1}{4} C_5^{-2} + \tfrac{1}{2} C_5 C_2 +\tfrac{1}{2} C_5^{-1} C_2+ \tfrac{1}{2} C_2 C_5 +\tfrac{1}{2} C_2 C_5^{-1}
\end{aligned}
\end{equation}
This, together with the exact values of $\lambda_\chi$ supplied in Table~\ref{Tb:2}, give the desired expressions of the fundamental models for the proper icosahedral group.

\begin{table}[h!]
\begin{center} 
\begin{tabular}{|c|c|c|c|c|c|}
\hline
\textbf{Irrep} & $A_g$ & $T_{1g}$ & $T_{2g}$ & $G_g$ & $H_g$ \\ \hline
$\mathbf{\lambda_{\chi}}$ & {\small  $-2 $} & {\small  $\frac{1}{4} \left(5-\sqrt{5}\right)$} &  {\small  $\frac{1}{4} \left(\sqrt{5}+5\right)$} &  {\small  $\frac{1}{4} \left(\sqrt{21}+1\right)$} &  {\small $\frac{1}{4} \left(\sqrt{13}+1\right)$} \\ \hline
\end{tabular}
\end{center}
 \caption{\small Selected eigenvalues of the adjacency operator $\Delta$ whose eigen-spaces sample the irreducible representations of the proper icosahedral group $I_p$.}
\label{Tb:2}
\end{table}

\begin{figure}[b!]
\label{fig:gapcrossing}
\centering
\hspace*{-1.15cm}
\includegraphics[scale=.175]{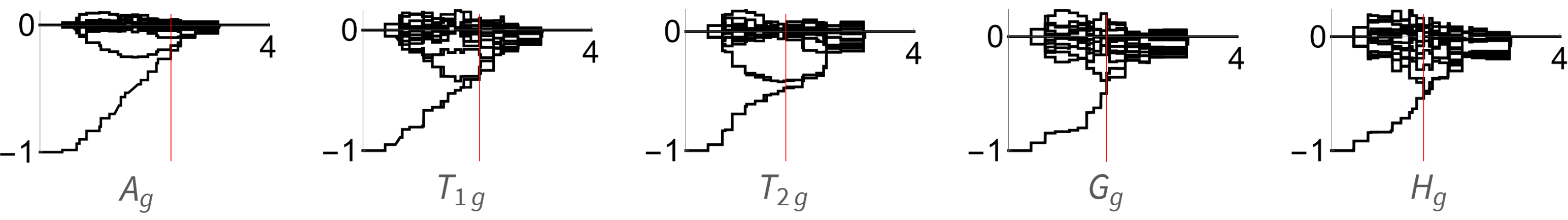}
\caption[short]{Spectrum of $h_{\chi}(t)$ from equation \eqref{eqn:deformation} as function of $t$, for the irreps in Table \ref{Tb:1}. The red line marks the closing of lowest spectral gap. }  
\end{figure}

The method explained above relies on the fact that the spectra of the symmetric operators over the Cayley graph is comprised entirely of discrete eigenvalues. This is always the case for finite groups, but, for infinite groups, the spectra consists of continuum bands. In such cases, we cannot make sense of the method we just presented. Given this fact, we want to explain a second method which works for both finite and infinite discrete groups. Consider the projection $p_\chi = \sum_{\gamma \in \Gamma}  \alpha_\gamma^{(\chi)} \, \gamma$ corresponding to one of the irreducible representations and a smooth cut-off function $\eta_\epsilon(x)$ over the real line, such that $\eta(x) = 1$ for $x<0$ and $\eta(x) =0$ for $x>\epsilon$. Here, $\epsilon$ is a small parameter, fixed at $10^{-3}$ in our simulations. Out of this data, we construct the self-adjoint element
\begin{equation}\label{eqn:deformation}
h_\chi(t) = - \sum_{\gamma \in \Gamma} \eta_\epsilon\big (t-d(e,\gamma)\big ) \, \alpha_\gamma^{(\chi)} \, \gamma \in \CM \Gamma,
\end{equation}
where $d(\gamma,\gamma')$ is the graph distance, which in this case is the great circle metric on the 3d Cayley graph of $I_p$. By varying the parameter $t$ from low to high values, we can continuously turn off coefficients in the standard presentation of $p_\chi$. Focusing now on the spectrum of $\pi_L(h_\chi(t))$ computed in Fig.~\ref{fig:gapcrossing}, we see that it consists of two degenerate eigenvalues $\{0,1\}$ for small values of $t$ because none of the coefficients are affected by the truncation. As $t$ is increased, coefficients are turned off and  the spectrum starts to flow. As long as the original spectral gap remains open, we can be sure that the spectral projection onto the lowest eigenvalue is stably homotopic with $\pi_L(p_\chi)$. As such, we can define the fundamental models to be $h_\chi(t_c + 0^+)$, where $t_c$ is the value of $t$ where the spectral gap closes. Fig.~\ref{Fig:HLoc} assesses the localization of the generated fundamental models.

\begin{figure}[t!]
\centering
\hspace*{-1.15cm}
\includegraphics[width=\textwidth]{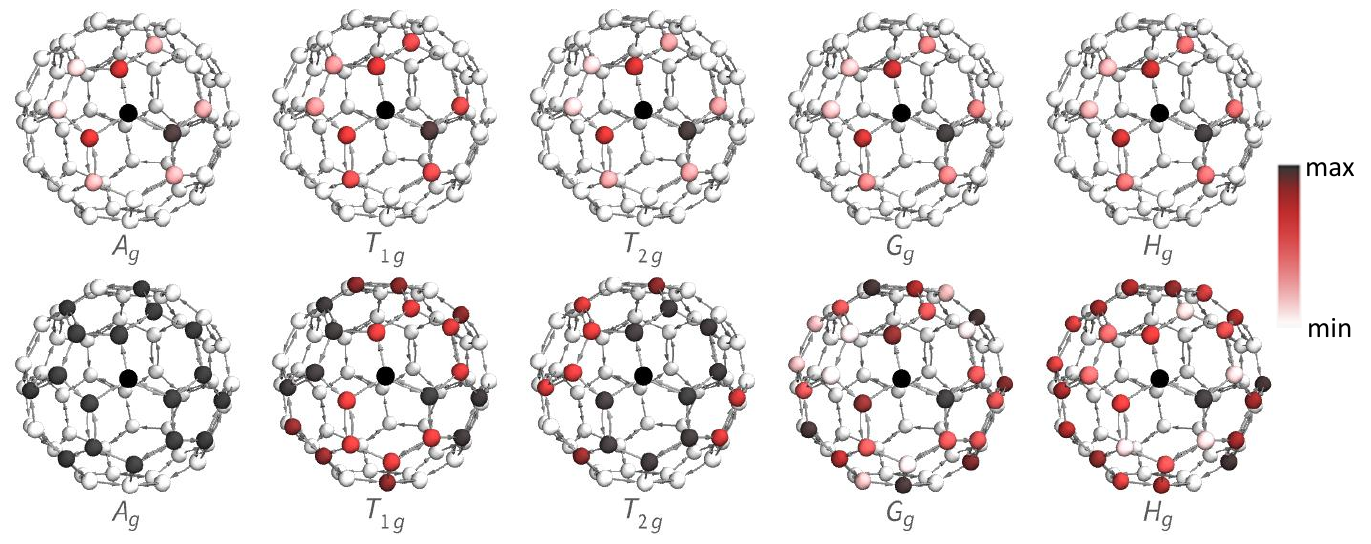}
\caption{(Top row) Coefficients $|\alpha_\gamma|$ of the model Hamiltonians $h_\chi$ from \eqref{eqn:toymodel}. (Bottom row) coefficients $|\alpha_\gamma|$ of the self-adjoint elements $h_\chi(t)$ from \eqref{eqn:deformation}, evaluated at $t$'s right before the lowest spectral gap crossing, as seen in Figure \ref{fig:gapcrossing}.}
\label{Fig:HLoc}
\end{figure}

\subsection{Spectral engineering}
\label{Sec:SpecEng}

\begin{figure}[b!]
    \centering
    \includegraphics[width = 0.4\linewidth]{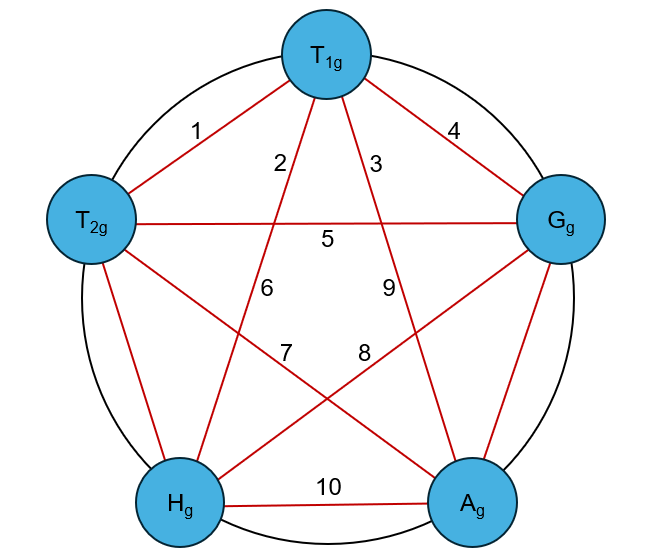}
    \caption{Labelling of the homotopies between the fundamental models $h_\chi$.}
    \label{Fig:LabellingScheme}
\end{figure}

In this subsection, we demonstrate that interpolations between pairs of distinct fundamental models generate topological spectral flows that are robust against perturbations. To start, we perturb $h_\chi$ from equation \eqref{eqn:toymodel} as
\begin{equation}
\label{eqn:compactpertubation}
    h_\chi + s\sum_{\gamma \in \Gamma} \tfrac{1}{2}(\alpha_\gamma + \bar \alpha_{\gamma^{-1}}) \gamma = h_\chi + s K_\chi, 
\end{equation}
where the complex coefficients $\alpha_\gamma$ were drawn randomly from the unit disk and $s$ is a parameter that controls the strength of the perturbation. From here on, all of the perturbed $h_\chi$'s will be rescaled so that their lowest spectral gap is unity. Since there are five topologically distinct fundamental models, we can generate ten distinct homotopies between pairs, as illustrated in Fig.~\ref{Fig:LabellingScheme}. The spectral flows generated by these interpolations are reported in Fig.~\ref{Fig:Interpolation1}. All spectral flows confirm the expected closings of lowest spectral gaps.

\begin{figure}[t!]
    \centering
    \includegraphics[width=0.9\linewidth]{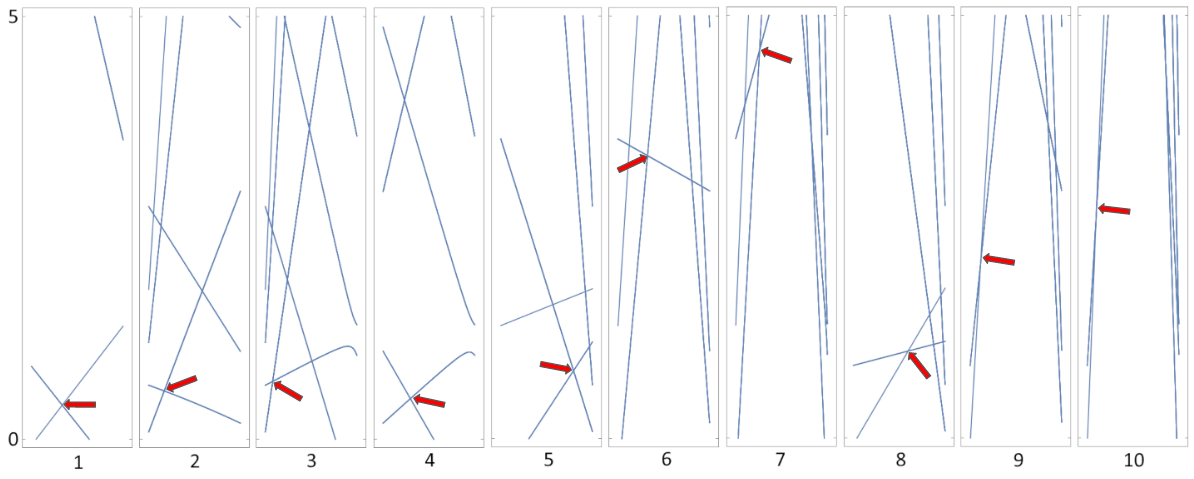}
    \caption{Spectrum of $\lambda (h_\chi + sK_\chi) + (1-\lambda) (h_{\chi'}+sK_{\chi'})$, where $h_\chi$ is as in equation \eqref{eqn:toymodel}. The graphs are labelled according to \ref{Fig:LabellingScheme} and zoomed in so that the changing of the lowest eigenvalue is visible. Here, $s$ is set equal to $0.1$ for all $h_\chi$. Red arrows indicate the first crossing of the two lowest eigenvalues.}
    \label{Fig:Interpolation1}
\end{figure}

To further illustrate the topological robustness of the spectral flows, we perturb $\pi_L(h_\chi)$, again from equation \eqref{eqn:toymodel}, with a diagonal operator $D_\chi$ on $\ell^2(\Gamma)$ with random entries, which represents an onsite disorder potential. Here, $D_\chi$ has been generated independently for each irrep $\chi$. Note that in this case, $\pi_L(h_\chi) + D_\chi$ is no longer in the group algebra, hence the degeneracy of the eigenvalues is broken. By sampling the disorder, the eigenvalues degenerate into bands of spectrum. The spectral flows generated by  interpolations between these disordered models are reported in Fig.~\ref{Fig:Interpolations2}. They again show closings of the lowest spectral gaps. 

\begin{figure}[t!]
    \centering
    \includegraphics[width=0.9\linewidth]{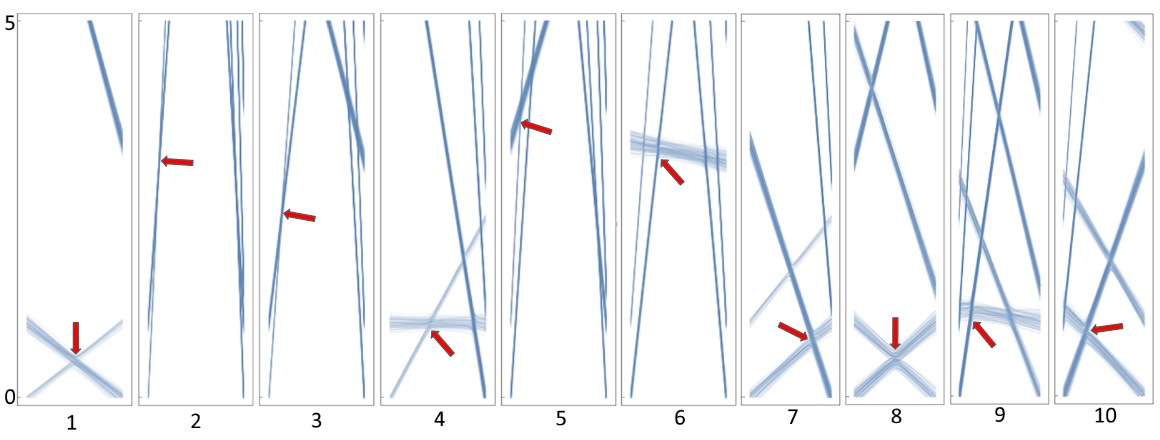}
    \caption{Spectrum of $\lambda (\pi_L(h_\chi) + D_\chi) + (1-\lambda) (\pi_L(h_{\chi'})+ D_{\chi'}) $, where $h_\chi$ is as in equation \eqref{eqn:toymodel}. The graphs are labelled according to \ref{Fig:LabellingScheme}, zoomed in so that the changing of the lowest spectral bands is visible. The elements of the random perturbations $D_\chi$ were sampled from a uniform distribution over $[-0.5,0.5]$. Red arrows indicate the first crossing of the two lowest eigenvalues.}
    \label{Fig:Interpolations2}
\end{figure}

\end{document}